\documentclass{amsart}
\usepackage{amsmath}%
\usepackage{amsfonts}%
\usepackage{amssymb}%
\usepackage{graphicx}
\newtheorem{theorem}{Theorem}
\theoremstyle{plain}

\newtheorem{corollary}{Corollary}

\numberwithin{equation}{section}
\begin{document}

\title[On the solution of Four Color Problem]{Properties of the dual planar triangulations}
\author{Natalia L. Malinina}
%\address[A. One and A. Two]
\email[Natalia L. Malinina]{malinina806@gmail.com}%
\thanks{This paper is a preliminary version of an article for some mathimatical journal}

\date{December 12, 2012}
\subjclass{Primary 05C10; Secondary 05C15} %
\keywords{Dual graphs, planar triangulation, four color problem}%
\dedicatory{Dedicated to my father, Leonid Malinin}

\begin{abstract}
This article is devoted to the properties of the planar triangulations. The conjugated planar triangulation will be introduced and on the base of the properties, which were achieved by the other authors there will be proved some theorems, which will show the properties of the dual triangulations. Also the numeric properties of the dual planar triangulations will be examined for the sake of understanding the interdependences of the cyclimatic numbers of different graphs between themselves. We'll see how the cyclomatic number of the planar conjugated triangulation depends on the cyclomatic number of the planar triangulation and how its increment depends on the number of the vertexes. These characteristics will be further very important for examining of Four Color Problem. The properties of the dual matrixes will also be examined. We will see that both matrixes on the one hand must meet the equal requirements, but on the other hand we will see that one characteristic cannot be fulfilled. This fact will further form the restrictions for the solution of Four Color Problem.

\end{abstract}
\maketitle

\section {Introduction: the basic properties of the planar graphs}
 
In real life there are many tasks, which depend on such topological characterization as either a possibility or impossibility to place the graph on the plane. So, let us begin with the brief review of the planar graph's properties, which were proved at various times by different authors.

\begin {theorem}
The boundaries of the different finite faces in the planar topological graph $G$ form the base of the independent cycles \cite {Berge}.
\end{theorem}

\begin {corollary}
For the planar map there exists the Euler formula \cite{Euler}:  $n-m+f=2$, where: $n$ --- is a number of the vertexes; $m$ --- is a number of the edges; $f$ --- is a number of the graph's faces. Indeed, the number of the finite faces is equal to the cyclomatic number $\nu$, so: $f=\nu+1=(m-n+1)=m-n+2$. 
And, finally: $n-m+f=2$. 
\end{corollary}

Many corollaries come from the Euler theorem, all of them were presented by Frank Harary \cite {Harary}.

\begin {corollary}
If $G$ is a planar $(n,m)$ map, in which every face presents a $n$--cycle, then: $m=p(n-2)/(p-2)$.
\end {corollary}

\begin {corollary}
If $G$ is a maximal planar $(n,m)$ graph, then its every face is a triangle and $m=3n-6$. If $G$ - is a planar graph, in which any face is 4 cycle, then $m=2n-4$, so, the number of the vertexes is even.
\end {corollary}

\begin {corollary}
If $G$ --- is an arbitrary planar $(n,m)$ graph and $n \ge 3$, then: $m \le 3n-6$. If graph $G$ is doubly connected and does not contain the triangles, then: $m \le 2n-4$.
\end {corollary}

\begin {corollary}
Every planar graph $G$ with $n \ge 4$ vertexes has at least $4$ vertexes, which degrees are not more then $5$.
\end {corollary}

\begin {theorem}
The graph is the maximal planar graph if and only if it is a planar triangulation \cite {Harary}.
\end {theorem}

\begin {theorem}
The graph is planar if and only if every its block is planar. It is evident that any planar graph can be lied (stacked) down on the sphere. The opposite is also true. It may be done with the help of different methods \cite {Harary}.
\end {theorem}

\begin {theorem}[Pontriagin--Kuratovskiy theorem] 
The necessary and sufficient conditions at which graph $G$ appears to be the planar graph consists in the fact that the graph cannot contain the partial subgraphs (fig. \ref {Fig:1}) of type either $1$ or $2$ \cite {Berge}.
\end {theorem}

\begin{figure}[htb]
		\includegraphics[width=0.5\textwidth]{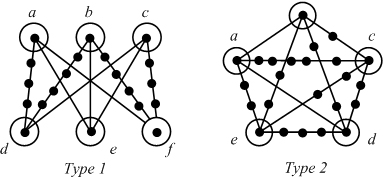}
	\caption{Partial subgraphs of two types}
	\label{Fig:1}
\end{figure}

The proof can be found in the book of Berge \cite {Berge}.

\begin {theorem}[Harary--Tutte theorem]
It is said that the graph $L=(X,U,P)$ satisfies the condition of Harary--Tutte, if it is impossible to turn it's skeleton into either the graph $F_5$, or into the graph $K_{3,3}$ (fig. \ref{Fig:1}) \cite {Harary}.
\end {theorem}

For more understanding of the behaviour of the dual planar graphs it will be necessary to examine their basic properties and also their numeric properties in order to determine the interrelations between them for the comparable graphs. The main of the properties are: the number of the vertexes, the number of the edges, the chromatic number, the cyclomatic number and the chromatic class. 

Also it will be very interesting to examine the properties of the dual matrixes in order to determine their form and size.

So, let us introduce the conjugated planar triangulation. 

\section{Properties of the planar triangulation and the conjugated planar triangulation}

Let us examine the planar triangulation and present some new theorems, because the solution of the Four Color Problem for the maximal graph $L$ means its solution for an arbitrary planar graph. 

It is well known that the planar triangulation possesses a number of important properties:

\begin{itemize}
\item 
The planar graph $L$ has $\nu(L)$ finite faces and one infinite face.
\item 
It is impossible to add to graph $L$ such an edge, which will not traverse some other edge, or which will not coincide with one or another already existing edge, or which will not cut off one more finite face from the infinite face.
\item 
The addition of the planar graph $M$ to the planar triangulation $L$ (fig. \ref{Fig:2}) by means of the new edge's insertion indicates the adding of the restrictions, which does not reduce the cyclomatic number of graph $L$ in comparison with the cyclomatic number of graph $M$.
\end{itemize}

And starting from this moment we'll begin to examine the planar triangulation together with its dual graph --- the conjugated planar triangulation.

Let us introduce the planar conjugated triangulation, which appears to be a graph dual to the planar triangulation. Graph $H(Q,R,\Phi)$ (fig. \ref{Fig:2}) appears to be the graph of the medians' faces of the triangulation $L$. Graph $H$ appears to be the first conversion of the initial graph $L$ \cite {Malinin}.

\begin{figure}[htb]
	\includegraphics[width=0.6\textwidth]{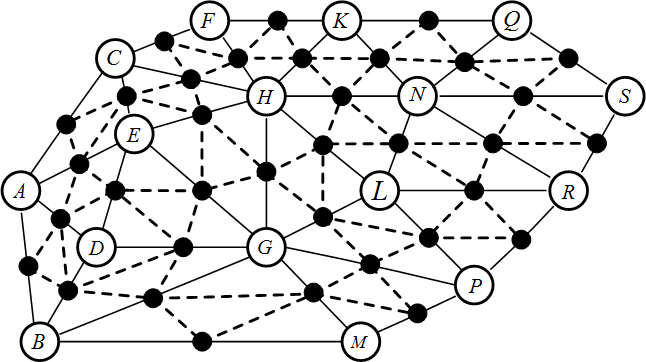}
	\caption{Graph $L$ --- solid line, graph $H$ --- dotted line}
	\label{Fig:2}
\end{figure}

For the purpose of constructing graph $H(Q,R,\Phi)$, or simply speaking, graph $H$, let us mark the centers of all the edges of graph $L$. We will accept these points as the vertexes $q_i$ of graph $H$. Let us connect them with the help of the medians of graph $L$ (fig. \ref{Fig:2}, \ref{Fig:3}). So, inside each face $f_l$ of the graph $L$ there will be created the triangle $f_h^{(1)}$. The triangles will be connected between themselves with the help of the vertexes (fig. \ref{Fig:3}).

\begin{figure}[htb]
	\includegraphics[width=0.6\textwidth]{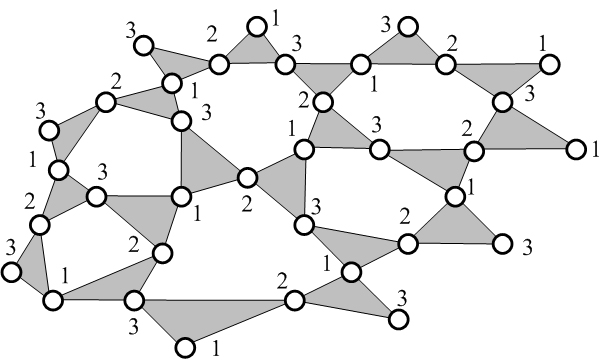}
	\caption{Graph $H$ --- the planar conjugated triangulation}
	\label{Fig:3}
\end{figure}

The construction of graph $H$ may be considered in another way. Let us construct graph $H$ as an adjacency edge graph of the $L$ graph. But in the received graph $H$ only those edges will remain, which are situated inside one $f_l$ face of graph $L$.

The received graph will be the sought--for graph $H$. 
The adjacency edge graphs are frequently named as conjugated graphs, that's why let us agree to name the graph $H$ as a planar conjugated triangulation. Graph $H$ presented in fig. \ref{Fig:2} with the help of the chain lines, and in fig. \ref{Fig:3} with the help of thin lines. The triangle faces of graph $H$ are filled (fig. \ref{Fig:3}).

Let us examine the properties of graph $H$ and prove them as theorems.

\begin{theorem}
The degrees of all the $H$ graph's vertexes are even and are equal either to 2 or 4.
\end {theorem}

\begin {proof}
There may be two cases (fig. \ref{Fig:2}, \ref{Fig:3}).
\begin{enumerate}
\item	
The graph's $H(V,Q)$ vertex $v_h$ belongs to one graph's $L$ face in case if the edge $v_h$ of graph $L(X,V)$, to which it is corresponding, is a part of the infinite face of graph $L$.
\item	
The graph's $H(V,Q)$ vertex $v_h$ belongs to two graph's $L$ face in case if the corresponding edge $v_h$ of the graph $L(X,V)$ divide the two finite faces of graph $L$.
\end{enumerate}
	
In the first case two medians of the finite face are incident to the vertex $v_h$; in another case four medians are incident to the vertex $v_h$. That is why the degrees of every graph's $H(V,Q)$ vertex $v_h$ can have only two values: either 2, or 4. The theorem is proved.
\end {proof}

The theorem can also be proved another way.

\begin {proof}
Each graph's $H(V,Q)$ vertex $v_h$ belongs either to one graph's $L$ face, if the edge $v_h$ of graph $L(X,V)$, to which it is corresponding, is also a part of the infinite face of graph $L$, or it belongs to two faces, if the corresponding edge divides two finite faces (fig. \ref{Fig:2}, \ref{Fig:3}). In the first case two medians of the finite face are incident to the vertex $v_h$; in another case four medians are incident to the vertex $v_h$. That is why the degrees of the graph's $H(V,Q)$ vertex $v_h$ can have only two values: either 2, or 4. The theorem is proved.
\end {proof}

\begin{theorem}
The planar triangulation $H$ contains the Euler circuit (contour), which goes exactly twice times along all the inner vertexes $q_i \in H$.
\end {theorem}

\begin {proof}
The proof evidently comes from the Euler theorem \cite {Berge}: graph possesses the Euler circuit if and only if it is connected and the number of the vertexes with the odd degrees is equal either to 0 or 2. 

All the vertexes of the planar conjugated triangulation have the even degrees of the vertexes. In addition, all the inner vertexes have the degree, equal to four. The theorem is proved. Let us denote such Euler circuit as the Bi--Euler circuit.
\end {proof}

\begin{corollary}
The edges $r_i \in H$ of the planar conjugated triangulation may be orientated in the direction of passing the Euler circuit.
\end {corollary}

\begin{theorem}
The finite graph's $H(V,Q)$ faces $f_h \in H$ put together two and only two subsets of the elements (or two classes of the equivalence):

\begin{enumerate}
\item	
The first subset of the faces, notably $\{f_h^{(1)}\}$, includes the faces, the boundary of which is made up from the medians of one face of graph $L(X,V)$. These faces are situated inside the faces $f_l \in L$. All these faces $f_h^{(1)}$ are the triangles.
\item	
The second subset of the faces, notably $\{f_h^{(2)}\}$ includes the faces, inside which the vertexes $q_i \in L$ of the initial graph are situated. All these faces have the polygon form with the number of the sides: $n \ge 3$.
\end{enumerate}
\end {theorem}

\begin {proof}
By the construction inside each face $f_l$ from $L(X,V)$ the face $f_h^{(1)}$ of graph $H$ is formed. These faces are formed by the medians of the $f_l$ faces.

If the vertex $x_i$ from $L(X,V)$ has the degree $\rho(x_i)$, then in this vertex the vertexes of the $\rho(x_i)$ faces of the $L(X,V)$ are concentrated. Inside each face $f_l$ of graph $L$ the face $f_h^{(1)}$ is situated. Its one side is directed to the $x_i$ vertex (the vertex of the initial graph $L$). The sides, which are gathered out of the faces $f_h^{(1)}$, generate the cycle, which form the boundary of the polygon face of graph $H$, inside which one vertex of graph $L(X,V)$ is situated. The number of the sides of the $f_h^{(1)}$ faces, which compose the self-contained cycle around the $x_i$ vertex is equal to the degree of the vertex, notably $\rho(L)$. So, these faces are $\rho$--angle; the number of the sides in each of them is equal to $\rho(L)$.

On the other hand, the planar conjugated triangulation has both the inner vertexes and the external vertexes. The external vertexes are those ones, which enter the boundary of the infinite face. The closed circuit out of the finite faces' medians cannot be generated around such vertexes (fig. \ref{Fig:4}). But the closed circuit out of the medians is simply generated around each inner vertex.

\begin{figure}[htp]
		\includegraphics[width=0.25\textwidth]{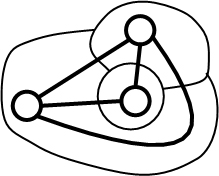}
	\caption{The arbitrary planar triangulation}
	\label{Fig:4}
\end{figure}

None of the inner vertexes of the triangulation $L$ can have the degree $\rho(x_i) \le 3$. That's why all the faces of the second subset $f_h^{(2)}$ have a number of the sides: $n_{\varphi}^{(2)} \ge 3$. The theorem is proved.
\end {proof}

\begin{theorem}
Two faces of the first subset $\{f_h^{(1)}\}$ can have one shared vertex, but cannot have shared edges (fig. \ref{Fig:3}).
\end {theorem}

\begin {proof}
Graph $L$ is the planar graph. Graph $H(V,Q)$ is also the planar graph. It has $f_h$ finite faces and one infinite face. Let us also denote the vertexes $v_h$, which enter the boundary of the infinite face as the external vertexes, and the other vertexes   as the inner ones.

Each face $f_{hi}^{(1)}$ of the first subset is situated in the face $f_{li}$ of graph $L$. That's why it cannot have the shared edge with neither the face $f_{hk}^{(1)}$ of graph $H$, which is situated inside any face $f_{ls}$ of graph $L$. But the face $f_{hi}^{(1)}$ can have the shared vertex with the face $f_{hi}^{(1)}$ if the corresponding faces of graph $L$ have the shared edge (as constructed). The theorem is proved.
\end {proof}

\begin{theorem}
Two faces of the second subset $\{f_h^{(2)}\}$ can have one shared vertex, but they can have no shared edges (fig. \ref{Fig:3}).
\end {theorem}

\begin {proof}
The validity of the theorem comes from the previous theorem. Indeed, the boundary of any face $f_{hk}^{(2)}$ of the second subset forms the sides of the faces $f_{hj}^{(1)}$ of the first subset. So, by no means two faces $f_{hk}^{(2)}$ and $f_{hj}^{(2)}$ can have a shared edge.

In each external vertex $v_hi$ of graph $H(V,Q)$, which has the degree $\rho(v_i )=2$, accurately meet a vertex of some face $f_h^{(1)}$ of the first subset and some vertex of the infinite face. In every external vertex $v_hi$, which has the degree $\rho(v_i )=4$, the vertexes of two faces of the first subset $\{f_{hi}^{(1)}\}$, of one face of the second subset $\{f_{hi}^{(2)}\}$ and of the infinite face $f_h^{\propto}$ are converged.

In each inner vertex $f_{hi}^{(2)}$ of graph $H(V,Q)$ exactly converge the vertexes of two faces of the $\{f_{hi}^{(1)}\}$ set and, therefore, of two faces of the $\{f_{hi}^{(2)}\}$ set. So, each inner vertex of graph $H(V,Q)$ is a shared vertex for some two faces of the second subset $\{f_{hi}^{(2)}\}$ of graph $H$. The theorem is proved.
\end {proof}

\begin{theorem}
Two faces $f_{hi}^{(1)}$ and $f_{hj}^{(2)}$ of the graph $H(V,Q)$ can share both one edge and two vertexes, or have none shared elements, but they cannot share only one vertex (fig. \ref{Fig:3}).
\end {theorem}

\begin {proof}
The validity of the theorem comes from the fact that the boundary of any face of the second subset $\{f_{hj}^{(2)}\}$ is formed by the sides (edges) of the faces $\{f_{hi}^{(1)}\}$ of the first subset. If some edge enters two faces, then two vertexes (the ends of the edge) enter the same faces. The theorem is proved.
\end {proof}

\begin{theorem}
The vertexes of the planar conjugated triangulation, which enter the boundary of any face $f_i \in \{f_{hi}^{(2)}\}$ always have the degree equal to $4$.
\end {theorem}

\begin {proof}
The boundary of any face $f_i \in \{f_{hi}^{(2)}\}$ of the second subset consists only of the sides of the faces $f_i \in \{f_{hi}^{(1)}\}$ of the first subset. Two vertexes are corresponding (or belong) to each of these sides. Besides two sides of the face $f_i \in \{f_{hi}^{(2)}\}$, two sides of the face $f_i \in \{f_{hi}^{(1)}\}$ enter every such vertex. So, all the vertexes, which enter the boundary of the face $f_i \in \{f_{hi}^{(2)}\}$ have the degree equal to $4$ (fig. \ref{Fig:3}). The theorem is proved.
\end {proof}

\begin{theorem}
Only the external vertexes of the planar conjugated triangulation have the degree equal to $2$.
\end {theorem}

\begin {proof}
Indeed, the degree $2$ can have only those vertexes of graph $H$, which correspond to those edges of graph $L$, which enter the boundary of the external face, i.e. the external vertexes.
\end {proof}

\begin{corollary}
The external vertexes of the planar conjugated triangulation can have the degrees $\rho(v_i)=2$ and $\rho(v_i)=4$. All the inner vertexes of the planar conjugated triangulation always have the degree $\rho(v_i)=4$.
\end {corollary}

\begin{corollary}
If graph $H$ has the cut nodes, then such vertexes always have the degree $\rho(v_i)=4$.
\end {corollary}

\begin {proof}
According to the construction.
\end {proof}

Also, it is seen by the construction, that graph $H$ always can be presented as the graph, which is constructed out of the triangle faces, connected with the help of the vertexes.

Our final target is the proof of the possibility of the coloring of graph's $H$ vertexes with three colors. Graph $H$ is the adjacency graph of the edges of the planar triangulation $L$. And if it is proved that graph's $H$ vertexes by all means can be colored with three colors, then Four Color Hypothesis will become the proved theorem.

Let us formulate this theorem.

\begin{theorem}
Let $H$ be the planar conjugated triangulation. Its chromatic class is: $\chi(H) \le 6$. It comes from the Brook's theorem. Then its chromatic number is: $\gamma(H) \le 3$.
\end {theorem}

Before beginning the proof of this theorem we ought to study many other properties of the planar triangulation. But first of all we will examine the properties of the special graphs. It will be made in the next article. The proof of the Four Color Problem appears to become the deeply prolonged work. 

\section{The numeric properties of dual planar triangulations} 

Let us examine the planar triangulation $L$ (fig. \ref{Fig:2}, thin line) and the planar conjugated triangulation $H$ (fig. \ref{Fig:2}, dotted line). Let us accept the following abridgement: the planar triangulation will be notes as PT; the planar conjugated triangulation will be noted as PCT. Let us introduce the following notations: $n$ --- is the number of the vertexes; $m$ --- is the number of the edges; $\mu$ --- is the number of the faces. Let us develop the dependences for both the number of the vertexes and the number of the edges of the planar conjugated triangulation from the number of vertexes of planar triangulation.

\begin{theorem}
If $m_L$ --- is a number of the edges (arcs), $n_L$ --- is a number of the vertexes of graph $L$, then for PCT (graph $H$) the dependences will be: $n_H=3n_L-6$ and $m_H=6n_L-12$.
\end{theorem}

\begin{proof}
According to the construction of graph $H$: $n_H=m_L$. The number of the finite faces in graph $H$: $\mu_H=\mu_H^{(1)}+\mu_H^{(2)}$, where $\mu_H^{(1)}$ --- is the number of the triangle faces in $f_h^{(1)}$ and $\mu_H^{(2)}$ --- is the number of the polygon faces in $f_h^{(2}$.

According to the construction we also have for graph $H$: $\mu_H^{(1)}=\mu_L$. That is, the number of the triangle faces $f_h^{(1)}$ in graph $H$ is equal to the number of all the faces $\mu_L$ of graph $L$ (according to the construction). Under theorem 8 the number of all the polygon faces in graph $H$, including the infinite face, is equal to the number of the faces in PT and makes up: $\mu_H^{(0)}=n_L$. At that the number of the finite faces is equal to: $\mu_H=n_L-1$. According to the construction of graph $H$, the number of the graph's edges is equal to: $m_H=3\mu_H^{(1)}=3\mu_L$.

Three vertexes $q_i$ of the graph $H$ are corresponding to each face $\mu_L$ of graph $L$. At the same time each graph's $H$ vertex is corresponding to two faces of graph $L$. That is why $2n_H=3\mu_L=3\mu_H^{(1)}$ or $n_H=3/2 \mu_L=3/2 \mu_H^{(1)}$ or $m_L=n_H=3/2 \mu_L=3/2 \mu_H^{(1)}$.

The cyclomatic number of graph $H$ is equal to the number of the finite faces: $\nu(H)=\mu_H^{(1)}+\mu_H^{(2)}$. At the same time: $\nu(H)=m_H-n_H+1$.

Taking into account the formulas above, we'll get:
$$\nu(H)=3\mu_L-3/2 \mu_L+1=3/2 \mu_L+1$$
Thereafter: 
$$\mu_H^{(1)}+\mu_H^{(2)}=3/2 \mu_L+1=3/2 \mu_H^{(1)}+1$$
And next: 
$$\mu_H^{(2)}=1/2 \mu_H^{(1)}+1$$
As far as the number of the finite faces of second subset is equal to: $\mu_H^{(2)}=n_L-1$ we'll get: $n_L=1/2 \mu_H^{(1)}+2$. 

Hence: $\mu_H^{(1)}=2n_L-4$; and further we'll get $m_H=6n_L-12$. 

Finally we have: 
$$n_H=m_L=3n_L-6$$ 
$$m_H=6n_L-12$$
\end{proof}

Let us summarize the obtained results in table (fig. \ref{Fig:5}). 

\begin{figure}[htb]
		\includegraphics[width=0.7\textwidth]{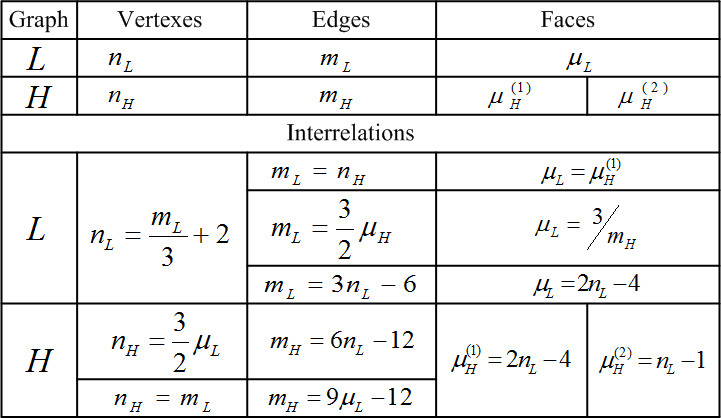}
	\caption{Dependencies between numeric properties for the different graphs}
	\label{Fig:5}
\end{figure}

\textit{Remark 1}: Some of these formulas can be found in the works of other authors, for example, Harary. But, as you can see, they also may be obtained applying the dual graphs.

\section {A cyclomatic number of the conjugated triangulation}

A cyclomatic number is always equal to the maximal number of the independent cycles in the graph. Each step, increasing the cyclomatic number, also increases the maximal number of the independent cycles. And this fact leads to the increasing of the graph's complexity.

\begin{theorem}
A cyclomatic number of the planar conjugated triangulation $H$ depends only on the number of the vertexes $n_H$ in graph $H$ and is equal to : $\nu(H)=n_H+1$.
\end {theorem}

\begin {proof}
According to the formulas $\nu(H)=m_H-n_H+1$ and $m_H=2n_H$ we have $\nu(H)=n_H+1$. The theorem is proved.
\end {proof}

\begin{corollary}
As far as the number of the edges $m_H$ in the planar conjugated triangulation $H$ is equal to: $m_H=2n_H$, then the dependence of the cyclomatic number on the number of the edges may be presented as: $\nu(H)=m_H/2+1$.
\end {corollary}

\begin{corollary}
$\nu(H)=m_L+1$
\end{corollary}

\begin{proof} 
Indeed, as far as $n_H=m_L$ (fig. \ref{Fig:2}), but $m_L$ --- is the number of the pairs of the graph's $L$ vertexes; so: $\nu(H)=n_H+1=m_L+1$.
\end{proof}

\textit{Remark 2}: The previous theorem and corollary are very important and specify the fact that the number of the logical equations, which connect the pairs of differently colored vertexes in graph $L$ between themselves, is equal to the number of the logical equations, which connect between themselves the differently colored vertexes in graph $H$.

\textit{Remark 3}: In each case the number of the equations is a unit more than it is necessary to have in the system of the linearly independent equations. Thus, the system appears to be an over determined one. The algorithms of obtaining the exact solution are absent for such equations' systems. Maybe such circumstances will help us to understand the reasons for the absence of the local algorithms intended for the vertexes' coloring in the planar triangulations in a general case. And maybe it will help us to understand actual reasons for the numerous attempts to solve the Four Color Problem with the help of the mathematical statistics' methods.

\section{The interdependencies between the numeric properties of the planar triangulations}

Let us examine the mutual relations between the cyclomatic numbers of different planar graphs, taking into account the planar graph, the planar triangulation (PT) and the planar conjugated triangulation (PCT).

So, at last we have the planar homogeneous graph $M$, the planar triangulation $L$ (generated out of graph $M$) and the planar conjugated triangulation $H$ (generated out of graph $L$). Let us prove some theorems.

\begin{theorem}
A planar triangulation's number of the edges is clarified only by the number of its vertexes.
\end {theorem}

\begin{proof}
From the formulas: $m_L=2/3 \mu_H^{(1)}$ and $\mu_H^{(1)}=2n_L-4$ we can receive: $m_L=3n_L-6$. Theorem is proved.
\end {proof}

This result can be obtained as the corollary from the Euler theorem \cite {Harary}. The identity of the obtained results allows to say that examining the planar graphs with the help of the dual graphs is correct and competent.

\begin{theorem}
A cyclomatic number of the planar triangulation is clarified only by the number of the vertexes of the planar triangulation.
\end {theorem}

\begin{proof}
Graph $L$ cyclomatic number: $\nu(L)=m_L-n_L+1$, but $m_L=3/2 \mu_{H_1}$, and $\mu_{H_1} =2n_L-4$. Therefore $m_L=3n_L-6$. Hence, $\nu(L)=3n_L-n_L-6+1=2n_L-5$. And finally: $\nu(L)=2n_L-5$. Theorem is proved.
\end{proof}

\textit{Remark 4}: These results also may be obtained from the Euler Theorem's corollary \cite {Harary}. And again the equivalence of the results confirms the competence of the examining planar graphs' properties with the help of the results, obtained for the dual graphs \cite {Malinin}, notably with the help of the planar conjugated triangulation.

Let us now deduce a series of the equations, which will connect the cyclomatic numbers of the planar graphs $L$ and $H$.

$\nu(L)=m_L-n_L+1$, but: $n_L=m_L/3+2$, so: $\nu(L)=2/3 m_L-1$.

On the other hand: $\nu(L)=2n_L-5$.

Then: $\nu(L)=2/3 m_L+4-5=2/3 m_L-1$.

As: $m_L=n_H$, so: $\nu(L)=2/3 n_{H-1}$.

In its turn $\nu(H)=n_H+1$ and $n_H=\nu(H)-1$.

Hence: 

$$\nu(L)=2/3 \nu(H)-2/3-1= 2/3 \nu(H)-5/3$$.

And next: 
$$\nu(H)=3/2 \nu(L)+5/2$$

Or: 
$$\nu(H)=\nu(L)+(\nu(L)+5)/2$$.

Or: $\nu(H)=\nu(L)+\Delta(L)$, where $\Delta(L)=(\nu(L)+5)/2$ (the increment of the cyclomatic number). From this fact it also comes that the increment of the cyclomatic number is always even. 

\textit{Remark 5}: This result shows us that at the conversion of PT into PCT the cyclomatic number increases, at that the increment of the cyclomatic number depends on the value of the initial cyclomatic number. As for the applications, we can say that the graph's complexity (a number of the independent cycles) at the process of the conversion into the conjugated graph increases. The more complicated the initial graph was, the bigger is the increments' growth. At the first sight it seems that the application of the properties of the dual graphs \cite {Malinina} to the examining of the properties of the planar triangulations may make the results worse. But step-by-step it will be shown that only the introduction of PCT permits us to widen considerably our view on the properties of PT.

It is known that the cyclomatic number is always integer. From the formula above, notably, $\nu(L)=(2\nu(H)-5)/3$ it comes that the cyclomatic number is multiple to 3. This dependence is presented in fig. \ref{Fig:6}.

\begin{figure}[htb]
	\includegraphics[width=0.7\textwidth]{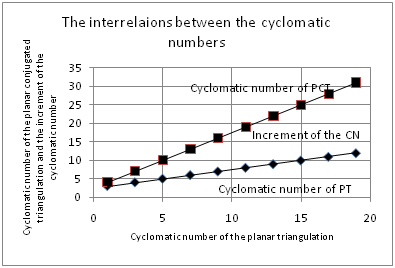}
	\caption{Interrelations between the cyclomatic numbers of the planar triangulation and the planar conjugated triangulation}
	\label{Fig:6}
\end{figure}

The graphical representation, of course, is somehow a little bit relative because the cyclomatic number is always even. This fact is displayed with the help of the markers. But it (fig. \ref{Fig:6}) again shows us that the cyclomatic number increases at the conversion from the PT to PCT. But in spite of the fact that the complexity increases, our chance to find the solution is not equal to zero. It also shows us that the triangulations exist only at some certain number of the graph's vertexes. Let us go further and see what we can do in such tremendously difficult circumstances.

For the cyclomatic number $\nu(H)$ of graph $H$ other dependences can be developed as well.

According to the definition: $\nu(H)=m_H-n_H+1$. 

As: $n_H=m_H/2$, so $\nu(H)=m_H-m_H/2+1=m_H/2+1$.

Let's find the relationship between the cyclomatic numbers for graphs $M$ and $L$.

$\nu(M)=\mu_M-1=n_L-1$, and $\nu(L)=2n_L-5$.
 
After some transformations we can receive that: $\nu(L)=2\nu(M)-3$.

Then $\nu(M)=(\nu(L)+3)/2$. 

And next: 
$$\nu(H)=2\nu(M)-3+(2\nu(M)-3+5)/2=2\nu(M)+\nu(M)-2$$.

Finally: 
$$\nu(H)=3\nu(M)-2$$ 

So: 
$$\nu(M)=\nu(H)/3+2/3=\frac{\nu(H)+2}{3}$$

Formulas, obtained for the determination of the cyclomatic numbers of the equivalent planar graphs are summarized in table (fig. \ref{Fig:8}). 

\begin{figure}[htb]
	\includegraphics[width=0.75\textwidth]{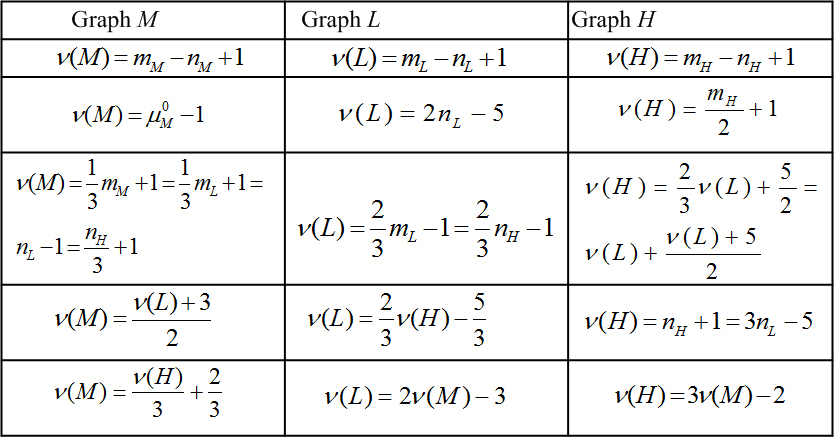}
  \caption{Formulas for the cyclomatic numbers of the different dual graphs}
	\label{Fig:8}
\end{figure}

Their graphical representation is presented in fig.(\ref{Fig:9} and \ref{Fig:10}. The upper line in fig.\ref{Fig:6} shows us the relationship between the number of the edges and the number of the vertexes of the planar triangulation. The lower line represents the relationship between the cyclomatic number and the number of the vertexes also of the planar triangulation. 

\begin{figure}[htb]
		\includegraphics[width=0.8\textwidth]{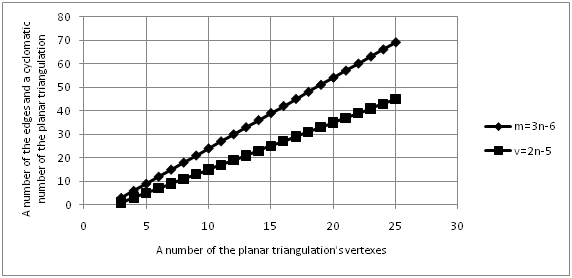}
	\caption{}
	\label{Fig:9}
\end{figure}

Similarly, in fig.\ref{Fig:10} such relations are presented for the planar conjugated triangulation. 
%Fig. 5 presents the relations between the cyclomatic number and the number of the vertexes (solid line for PT; dotted line for PCT).

\begin{figure}[h]
		\includegraphics[width=0.8\textwidth]{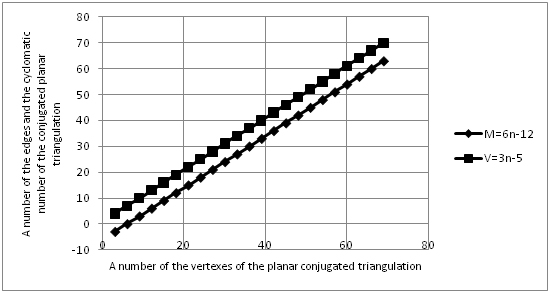}
	\caption{}
	\label{Fig:10}
\end{figure}

\section{Properties of matrixes $R$ and $F$ for graph $H$}

As far as graph $H$ contains the double Euler circuit (see above), all the edges of graph $H$ can be oriented in the direction of this cycle. As a result, graph $H$ (fig. \ref{Fig:3}) can be transformed into the direct graph and matrix $F$ can be composed as an adjacency matrix of its vertexes \cite {Malinin}. The straight converting gives us a possibility to compose $R$ matrix as an adjacency matrix of graph's $H$ edges.

So, let us examine the properties of matrixes both $R$ and $F$ of graph $H$.
Then we'll examine the change of the cyclomatic number at the matrixes' $F$ straight converting.

We know that matrix $R$ is an adjacency matrix of graph's $H$ edges; and matrix $F$ is the adjacency matrix of its vertexes. The properties of both of them can be proved in the following theorems.

\begin{theorem}
Matrix $R$ is anti--symmetric one.
\end {theorem}

\begin {proof}
Under the fact of the absence of the contours, which have the length $2$ (according to the construction) in graph $H$.
\end {proof}

\begin {theorem}
Matrix $F$ is anti--symmetric one.
\end {theorem}

\begin {proof}
The deduction is by analogy with the previous theorem.
\end {proof}

\begin {theorem}
Matrix $R$ is the quasicanonical one.
\end {theorem}

\begin {proof}
Indeed, as all the vertexes $q_i$ of graph $H$ have two by two both the entering and the coming out edges. So, all the elements $r_{ij}=1$ of matrix $R$ are distributed properly among the non-intersecting submatrixes $|r_{ij} |_k^p$, which size is $2\times 2$ \cite {Malinin}.
\end {proof}

\begin {corollary}
The total sum of the elements $r_{ij}=1$ of matrix $R$ is multiple to $4$, that is: $\sum_{i=1}^{m_H}\sum_{j=1}^{m_H}r_{ij}=4k$, where $k=3,4,5,...K$, and $K$ --- is the arbitrary large integer number.
\end {corollary}

\begin {theorem} \label{MatrixRedges}
The number of the edges $m_H$ of graph $H$ and also the number of the rows in matrix $R$ is multiple to $6$.
\end {theorem}

\begin {proof}
As graph $H$ is constructed from the triangles (it has no edges, which do not enter any triangle), which are connected with the help of the vertexes, the number of the edges of graph $H$ is equal to: $m_H=3\mu_H^{(1)}$. That's why the number of graph's $H$ edges $m_H$ is always multiple to 3. Next, in each matrix's $R$ row we have the exactly two elements $r_{ij}=1$, according to the presence of the double Euler circuit in graph $H$. That's why it is evident, that the total sum of the matrix's elements is equal to: $\sum_{i=1}^{m_H}\sum_{j=1}^{m_H}r_{ij}=2m_H=6\mu_H^{(1)}$, or it is multiple to $6$.

In compliance with the corollary from the theorem on the quasicanonicality of matrix $R$ the total sum of all the elements in it is equal to: $\sum_{i=1}^{m_H}\sum_{j=1}^{m_H}r_{ij}=4k$, so it is multiple to $4$. From the formulas above it also follows that the total sum must be equal to: $\sum_{i=1}^{m_H}\sum_{j=1}^{m_H}r_{ij}=12k$, where $k=1,2,...K$. So, the number of the matrix's $R$ rows must be multiple to $6$.

Thus: $m_H=6k_1$, where $k_1=1,2,...K$
\end {proof}

There may also be another proof of this theorem.

\begin {proof}
From $m_H=3\mu_{H_1 }=3\mu_L$ and $m_H=6n_L-12$ (see above) it follows that: $3\mu_{H_1 }=6n_L-12$. It means that: $\mu_{H_1 }=2(n_L-2)$ --- is always even. So: $m_H=3\mu_{H_1}$ - is multiple to $6$, that is: $m_H=6\cdot k_2$, where $k_2=1,2,...,K$.
\end {proof}

\begin {theorem} \label{MatrixFrows}
The number $n_H$ of matrix $F$ rows is multiple to $3$.
\end {theorem}

\begin {proof}
It is known that: $2n_H=3\mu_H^{(1)}$. It was proved above that $\mu_H^{(1)}$ is even, notably, $\mu_H^{(1)}=2k_3$, where $k_3$ --- is integer.
 
This implies that $2n_H=6k_3$ and the number of the graph's $H$ vertexes is multiple to $3$. The theorem is proved.
\end {proof}

The formulas above permit us to compose a table of the possible numbers of the values: $n_L$, $n_H$, $m_H$ and $\nu(H)$ relatively to some hypothetical integer number $k$. We have $3k_3=3n_L-6$, that is: $n_L=k_3+2$, where $k_3=0,1,2,...K$. Let us compose a table (fig. \ref {Fig:11}), but replace all the coefficients by a single one $(k_1\rightarrow k_2\rightarrow k)$.

\begin{figure}[htb]
		\includegraphics[width=0.8\textwidth]{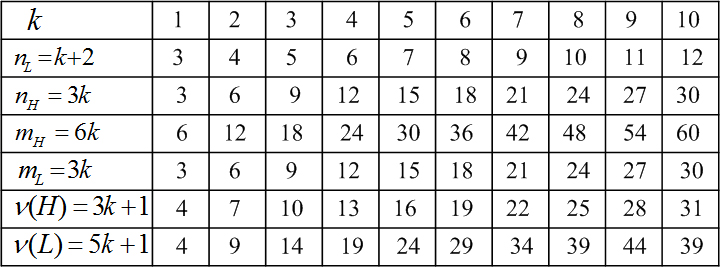}
	\caption{}
	\label{Fig:11}
\end{figure}

So, in proportion to the value of coefficient $k$ we can possess only such triangulations as it is presented in table (fig. \ref{Fig:11}). We can also construct the graphic portrayals of some useful functions. Fig. \ref {Fig:12} presents the relationship of the comparative number of the logical equations for one logical variable from the number of the vertexes in the planar triangulation. The function $f(n_L )=\nu(L)/n_L$ is presented by the solid line; the function $\varphi(n_H)=\nu(H)/n_H$ is presented by the dotted line. 

Fig. \ref {Fig:12} shows that the equations set is an over determined one for graph $L$, beginning from the number of the vertexes equal to $5$. If the number of the vertexes is less than $5$, then the equations set is the under-determined one. For graph $H$ the equations set is always over-determined, but the degree of the over-determination comes down with the increasing of the number of vertexes.

\begin{figure}[htb]
		\includegraphics[width=0.9\textwidth]{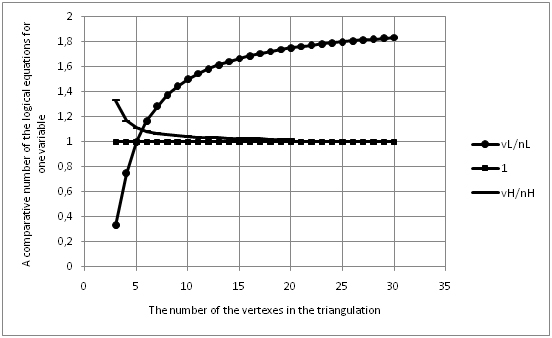}
	\caption{The relations between the cyclomatic numbers and numbers of the vertexes for the planar triangulation and the planar conjugated triangulation}
	\label{Fig:12}
\end{figure}

The relations of the vertexes' number ($n_L$ --- solid line), edges' number ($m_L$ --- dotted line) and the cyclomatic number ($\nu(L)$ --- chain line) from coefficient $k$ for the planar triangulation are presented in fig. \ref{Fig:13}. 

\begin{figure}[htb]
	\includegraphics[width=0.7\textwidth]{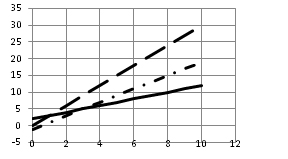}
	\caption{}
	\label{Fig:13}
\end{figure}

The relations of the vertexes' number ($n_L$ --- solid line), edges' number ($m_L$ --- dotted line) and the cyclomatic number ($\nu(L)$ --- chain line) from coefficient $k$ for the planar conjugated triangulation are presented in fig. \ref {Fig:14}.

\begin{figure}[htb]
		\includegraphics[width=0.7\textwidth]{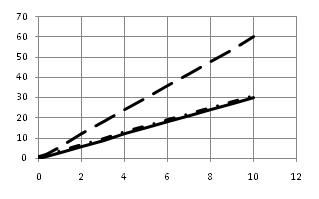}
	\caption{}
	\label{Fig:14}
\end{figure}

Let us examine some cases, based on the data from table in fig. \ref {Fig:11}.

\subsection {Case 1}

The initial data: $k=1$; $n_L=3$; $n_H=3$; $m_H=6$; $\nu(H)=4$. 
This case (fig. \ref{Fig:15}) corresponds to the map of three countries, for which the Four Color Problem is meaningless. The case appears to be the peculiar one, because:

\begin{enumerate}
\item	
The degrees of the vertexes are: $\rho(v_i )=2$.
\item	
The conjugated triangulation contains cycles equal to $2$.
\end {enumerate}

\begin{figure}[htb]
		\includegraphics[width=0.9\textwidth]{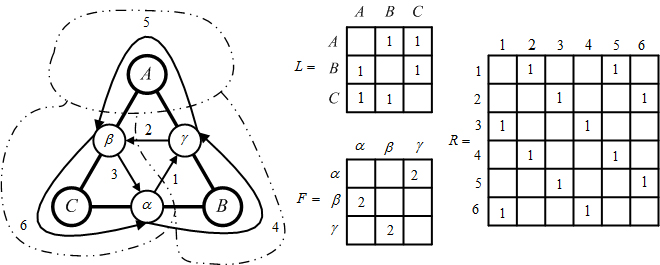}
	\caption{}
	\label{Fig:15}
\end{figure}

\subsection {Case 2}

Initial data: $k_2=2$; $n_L=4$; $n_H=6$; $m_H=12$; $\nu(H)=7$; $\sum_{i=1}^{m_H}\sum_{j=1}^{m_H}r_{ij}=24$.  
It is a case, which corresponds to the map of $4$ countries. The triangulation is a tetrahedron one. The task is knowingly solvable and is presented in fig. \ref{Fig:16}.

\begin{figure}[htb]
	\includegraphics[width=1\textwidth]{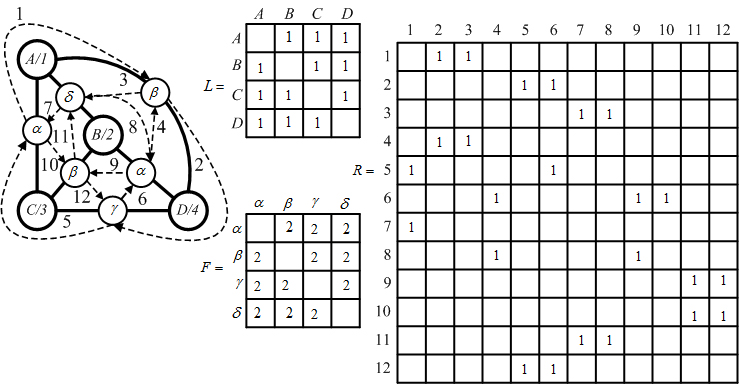}
	\caption{}
	\label{Fig:16}
\end{figure}

\subsection {Case 3}

Initial data: $k_2 \ge 3$. 
Beginning from this case, all the possible variants correspond to the Four Color Problem and may be formally solved only if Four Color Problem is the equitable one.

\section {The compactness of the matrixes' $F$ and $R$ filling}

Let us examine the problem of matrix's $F$ compactness filling ($F$ --- is an adjacency graph's $H$ matrix of the vertexes). The total number of the elements in such a matrix is equal to the doubled number of the vertexes:
 $$\sum_{i=1}^{n_H}\sum_{j=1}^{n_H}r_{ij}=2n_H$$.

In this case the matrix's compactness filling will be equal to: $\sigma_F=(2n_H)/(n_H )^2 =2/n_H $.
The number of the rows, which fit to one element in the row, may be presented as:  $1/\sigma_F =n_H/2$.

Thus, for the principal possibility to compose matrix $F$, taking into account its anti-symmetry, each matrix's element ought to correspond to not less than two cells, excluding the diagonal cells, that is: $((n_H )^2-n_H)/(2\cdot 2\cdot n_H )\ge 1$, beginning from the number of the vertexes $n_H=3$. So, we have: $(n_H )^2-n_H\ge 4 n_H$ or $n_H\ge 5 n_H$.

And finally: $n_H\ge5$.

Thus, matrix $F$ must have the number of rows, not less than $5$. Under the condition of theorems \ref{MatrixRedges} and \ref{MatrixFrows} (a condition of the multiplication factor: the minimal number of rows must be equal to $6$ (multiplication factor to $3$). Thus, case $1$, mentioned above, must be excluded. So, from the obtained result and the formula: $m_H=2n_H$, it follows that the minimal number of the rows in matrix $F$ must be equal to: $m_{H_{min}}=12$.

Let us examine the problem of the filling compactness of matrix $R$, which is the adjacency matrix of graph's $H$ edges.

The total amount of the elements in matrix $R$ can be expressed with the formula: $$\sum_{i=1}^{m_H}\sum_{j=1}^{m_H}r_{ij}=2m_H$$.

The compactness of the matrix' filling may be represented in the following way: $\sigma_R=(2m_H)/(m_H )^2 =2/m_H$. The reciprocal value will be: $1/\sigma_R =m_H/2=\lambda$, where: $\lambda$ --- is the number of the cells, fitting to one element.

So, for the principal possibility to compose matrix $R$, under the condition of its anti-symmetry, not less than two cells must fit to one element, excluding the diagonal cells. 

Similarly to the antecedent, we can get the general conditions, which are necessary for the composing of matrix $R$: $(m_H^2-m_H)/(2\cdot 2\cdot m_H )\ge 1$ or $m_H^2-m_H\ge 4 m_H$.

Finally we also have: $m_H\ge 5$.

Thus, the anti-symmetric matrix $R$, which meet the above-listed requirements, must have a number of the rows not less than $5$. It superfluously confirms the fact that the number of the rows must be multiple to $6$. In addition we know that the number of the rows must be more than $6$.

Thus, both matrixes $R$ and $F$, according to the number of the elements and the compactness of the filling meet the same requirements. But there exist some differences.
 
Matrix $R$ is always the quasicanonical one, because it is the adjacency edge matrix, and it can be converted with the help of the reverse converting  \cite {Malinin} into matrix $F$.

As opposed to it, matrix $F$ always appears to be the arbitrary matrix, and, as a rule, in a general case does not appear to be the quasi-canonical one as an adjacency matrix. Thus, as a rule, matrix $F$ cannot be considered as an adjacency edge matrix of some graph $G$, for which graph $H$ is a conjugated one.

Table in fig. \ref{Fig:17} shows us the relations between both the graphs and their conversions; also it shows the fact of either absence or existence of either Euler or Hamilton circuits in the graphs. 

\begin{figure}[htb]
		\includegraphics[width=0.9\textwidth]{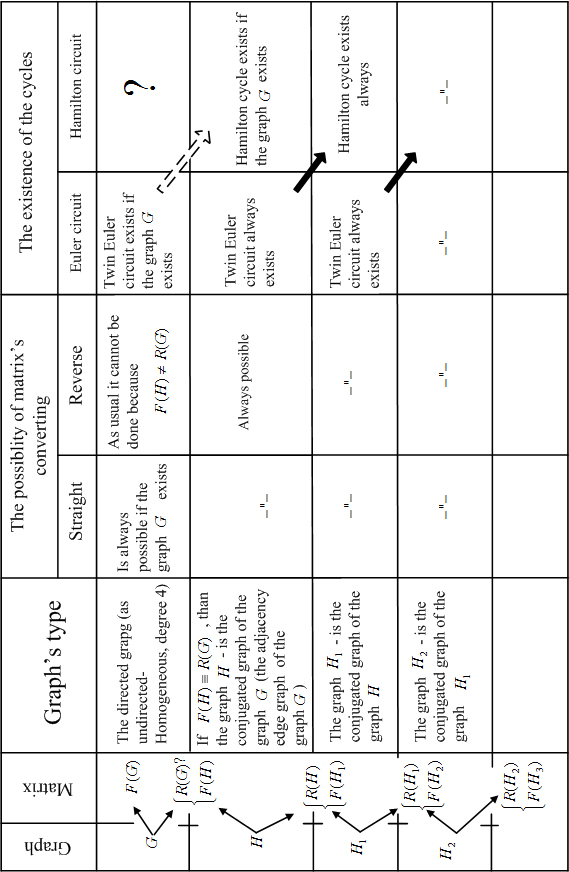}
	\caption{}
	\label{Fig:17}
\end{figure}

\section {The requirements for the matrix's $F$ properties}

Let us examine what requirements must meet matrix $F(H)$, in order to be considered as matrix $R(G)$.

Each submatrix, which have the size $2\times2$, in matrix $R(G)$ corresponds in graph $G \equiv H_1$ to one cycle of the length equal to $4$, if we consider this matrix as matrix $F(H_1)$\footnote{Graph $H_1$ is the first straight conversion of graph $H$ \cite{Malinin}.}. Thus, for the condition $F(H)\equiv R(G)$ to be fair, it is necessary for graph $H$ to have the number of the cycles with the length equal to $4$ not less than $n_H/2$. This condition also means that graph $L$, to which graph $G$ is corresponding, must have the number $n_L^{(4)}$ of the vertexes with the degree equal to $4$ is not less than $n_L/2$. The maximal number of the vertexes $n_L^{(4)}$ cannot be more than $n_L$. 

But
$$n_L=\frac{n_H+6}{3}=\frac{n_H}{3}+2$$.

Let us find the minimal number $n_H$, at which the existence of the quasi-canonical matrix $F(H) \equiv R(G)$ is possible.

We know that
 $$\frac{(n_H)_{max}+6}{3}=\frac{(n_H)_{max}}{2}$$. 

We can also write: 
$$2(n_H)_{max}+12=3(n_H)_{max}$$. 

So, we can get: $(n_H)_{max}=12$.

And finally: 
$$(n_L)_{max}=\frac{12+6}{3}=6$$.

Thus, only the adjacency matrixes of either the tetrahedron, or the quinquecuspidate polyhedron, or the octahedron can be the adjacency matrixes of graph's $G$ edges.

So, the very important conclusion follows that not a single adjacency vertexes' matrix $F(H)$, which has the number of the vertexes of the initial triangulation more than $7$ $(n_L \ge 7)$ can be subjected to the operation of the reverse converting.

We must draw special attention to this conclusion, because exactly this one will be needed at the examining the necessary and sufficient conditions of the right coloring of the graph into four colors.

Graph $H$ always has a double (twin) Euler circuit. So, graph $H_1$ (the straight conversion of graph $H$) always has Hamilton cycle \cite {Malinin}. But for the existence of Hamilton cycle in graph $H$ it is sufficient for graph $G$ to have a double (twin) Euler circuit. It is possible only if graph $G$ exists as an inverse conversion of graph $H$.

But graph $G$ cannot exist as an inverse conversion at the number of the vertexes $n_L \ge 7$. So, it is impossible to prove the existence of Hamilton cycle in graph $H$, relying on the operation of the inverse conversion of graph $H$. Thus, it is impossible in the limits of the existing axiomatic system.

Meanwhile, it will be useful to recollect Hedel's theorem on the incompleteness \cite {Hedel}. It is proved there that any formal axiomatic system contains the unsolvable hypotheses or that in any sufficiently complicated non contradictory theory a statement exists, which cannot be either proved or disproved by means of this theory.

The meaning of Hedel's theorem is also in the statement that any sufficiently strong formal non contradictory logic-mathematical system obligatory contains either the formula or the statement, which cannot be either proved or disproved in this system, but which is true. But this fact may be shown only with the help of some other methods, which go beyond the scope of the existing system. 

Hedel's theorem is applicable not only to the arithmetic, but also to every science, including philosophy. The global meaning of Hedel's theorem is in the fact that every science has problems, which cannot be solved beyond the scope of this science, but only in the scopes of a new more general extended theory. The methodological significance of Hedel's theorem consists in the fact that any fundamental theory is either the contradictory one, or is insufficient for the solution of some problems, which are arising in it. The appearing of the solution of the Four Color Problem with the help of the computer may serve one of the possible withdrawals for the Four Color Problem.

So, let us make first, sufficiently wary suggestion that the Four Color Problem is the statement, which we perhaps will not be able to prove in theorems because of Hedel. But, though the Four Color Problem contains the affirmation on its improvability, however, it is true. On the one hand, the computer proofs exist. On the other hand, over a period of one hundred and sixty years no contrary instances were found.

Let us suggest that the Four Color Problem is the true proposition, in spite of all the attempts to prove it with the help of the existing axiomatic system. Hence a quite acceptable conclusion follows: the existing axiomatic system is incomplete. It seems most probably that we must do the following: we ought to accept the Four Color Problem as an axiom and thus enlarge the existing list of the today's axiomatic system. The system will appear to become a complete one. We remind, that it is only a suggestion.

\section {Conclusions}

\begin{enumerate}
\item	
The properties of the planar graph, being proved earlier by other authors were examined. These properties will make the base for the theorems, which show the properties of the planar conjugated triangulation.
\item	
The graph, which is dual to the planar triangulation, is introduced.
\item	
Some theorems are proved, which shows the first properties of the planar conjugated triangulation.

\begin {itemize}
\item	
Theorem on the vertex's degrees.
\item	
Theorem on the existence of Euler circuit.
\item	
Theorem on two classes of the planar conjugated triangulation's faces.
\item	
Theorems on the properties of both the faces and the vertexes of the planar conjugated triangulation.
\end{itemize}
	
\item	
The Formulation of the theorem on both the chromatic class and chromatic number of the planar conjugated triangulation.

\item	
The relation of the planar conjugated triangulation edge's number depends on the vertexes' number of the planar triangulation.
	
\item	
The relation of the planar conjugated triangulation vertexes' number depends on the faces' number of the planar triangulation.
	
\item	
The relation of the planar conjugated triangulation faces depends on the vertexes' number of the planar triangulation.
	
\item	
The cyclomatic number of the planar conjugated triangulation depends only on its vertexes' number or on the edges' number of the planar triangulation.

\item 
The relations between the numeric properties of the planar graphs are received, notably: the homogeneous planar graph $M$, the planar triangulation $L$ and the planar conjugated triangulation $H$.

\item 
It was demonstrated that at the transformation of the planar triangulation $L$ into the planar conjugated triangulation $H$ the cyclomatic number increases and its increment depends on the cyclomatic number of the planar triangulation, in addition it is always even.

\item 
The numeric properties were identified and estimated.

\begin {itemize}
\item	
The number of all the finite faces of graph $H$: $\mu_H=\mu_H^{(1)}+\mu_H^{(2)}$, where: $\mu_H^{(1)}=\mu_L$ and $\mu_H^{(2)}=n_L-1$.
\item	
The number of the vertexes of graph $H$: $n_H=3n_L-6$ or $n_H=m_H/2$.
\item	
The number of the edges of graph $H$: $m_H=6n_L-12$ or $m_H=2n_H$.
\item	
The cyclomatic number of graph $H$ can be expressed with the following formulas:

\begin {itemize}
\item 
$\nu(H)=\mu_H=\mu_H^{(1)}+\mu_H^{(2)}$
\item 
$\nu(H)=3(\mu_H^{(1)})/2 +1=3(\mu_L)/2 +1$
\item 
$\nu(H)=3(\nu(L))/2 +5/2=\nu(L)$
\item 
$\nu(H)=3n_L-5$
\end {itemize}
\end {itemize}

\item 
The adjacency edge matrix $R$ of graph $H$ possesses the following characteristics:

\begin {itemize} 
\item 
Matrix $R$ --- is anti-symmetric one.
\item
The number of the rows in matrix $R$ is multiple to $6$. That is: $m_H=6k_1$, where $k_1$ --- integer. For the real cases $m_{H_{min}}=12$, that is: $k_1=2,3,4...$.
\item 
The mean number of the cells corresponding to one element in the row is: $1/\sigma_R =m_H/2$.
\item 
The number of the elements in each row is equal to $2$.
\item 
The total sum of all the elements in matrix $R$: $\sum_{i=1}^{m_H} \sum_{j=1}^{m_H}r_{ij}=2m_H$.
\item 
All the elements $r_{ij}=1$ of matrix $R$ are distributed among the submatrixes $\left\|r_{ij}\right\|_1^n$, which size is $2 \times 2$, that is why matrix $R$ --- is a quasi-canonical one.
\end {itemize}

\item 
The adjacency vertex matrix $F$ of graph $H$ possesses the following characteristics:

\begin {itemize}
\item 
Matrix $F$ --- is anti-symmetric one.
\item 
The number of the rows $n_H$ in matrix $F$ is equal to $3$, and the minimal number of the rows for a real case is: $n_{H_{min} }=6$.
\item 
The mean number of the cells corresponding to one element in the row is: $1/\sigma_F =n_H/2$.
\item 
The number of the elements in each row is equal to $2$. 
\item 
The total sum of all the elements in matrix $F$: $\sum_{i=1}^{n_H}\sum_{j=1}^{n_H}r_{ij}=2n_H$.
\item
None of matrixes $F(H)$ at $n_H>12$ (or at $n_L>6$) can be a quasi-canonical adjacency matrix of graph's $G$ edges.
\end {itemize}

\item 
At $n_H>12$ none of matrixes $F(H)$ can be exposed to the operation of the reverse conversion. 
\end{enumerate}

As a finall result: it is impossible to construct graph $G$ as the triangulation $G$ at $n_G>6$ and also at the condition of the existing of the double Euler circuit.

Here also comes the impossibility to prove the existence of the Hamilton circuit in graph $H$ on the base of the properties of the dual graphs, that is: on the base of the existing axiomatic system. 

So, we again come into the collision with the first possible proof on the impossibility to solve Four Color Problem. Let us try from the other side and examine the characteristics of the planar triangulations more thoroughly. Let us suggest that Four Color Problem is unprovable. But the proofs with the help of the computer also come into collision with enormous computing problems at the solution of the Four Color Problem. The situation in the following computer proofs became much better after the proof of Appel and Heiken was done. But nevertheless the situation with the computer proofs is not that transparent for us to verify the mathematical certainty of the computer proof.

Well, what extraordinary characteristics of the planar graphs are preventing us from learning the truth? 

In search of the answer let us examine the special properties of the planar triangulations, but not immediately. In the next article we'll try to examine the properties of some special graphs.

\section {Acknowkledgments}

A lot of thanks for those members of my family who undergone all the difficulties side by side with me and who encouraged me in my work.  

Also I will be very grateful to those readers, who will find and send me a word about the uncovered misprints or some errors in order to improve the text.

\end{document}